\let\oldvec\vec% Store \vec in \oldvec
\documentclass[runningheads]{llncs}
\let\vec\oldvec% Restore \vec from \oldvec
\usepackage{graphicx,amsmath,amsfonts}
\usepackage[utf8]{inputenc}
\usepackage[T1]{fontenc}
\usepackage{algorithm}
\usepackage[noend]{algpseudocode}
\usepackage{amssymb}
\usepackage{bm}
\usepackage{float}
\usepackage{pdfpages}
\usepackage{xfrac}

\usepackage{enumerate}
\usepackage{enumitem}
\usepackage{cite}

\renewcommand{\iff}{\Leftrightarrow}

\begin{document}
\title{Payment Networks as Creation Games}
\author{Georgia Avarikioti \and
Rolf Scheuner \and Roger Wattenhofer}
\authorrunning{G. Avarikioti et al.}
\institute{ETH Zurich, Switzerland
%\email{zetavar@ethz.ch}\\
%\url{http://www.springer.com/gp/computer-science/lncs}
\\
\email{\{zetavar,schrolf,wattenhofer\}@ethz.ch}}
\maketitle              % typeset the header of the contribution
%
%------------------------------------------------
%---------------ABSTRACT------------------------
%-----------------------------------------------
\begin{abstract}
Payment networks were introduced to address the limitation on the transaction throughput of popular blockchains. 
To open a payment channel one has to publish a transaction on-chain and pay the appropriate transaction fee. 
A transaction can be routed in the network, as long as there is a path of channels with the necessary capital.
The intermediate nodes on this path can ask for a fee to forward the transaction.
Hence, opening channels, although costly, can benefit a party, both by reducing the cost of the party for sending a transaction and by collecting the fees from forwarding transactions of other parties. 

This trade-off spawns a network creation game between the channel parties. In this work, we introduce the first game theoretic model for analyzing the network creation game on blockchain payment channels.
Further, we examine various network structures (path, star, complete bipartite graph and clique) and determine for each one of them the constraints (fee value) under which they constitute a Nash equilibrium, given a fixed fee policy.
Last, we show that the star is a Nash equilibrium when each channel party can freely decide the channel fee. On the other hand, we prove the complete bipartite graph can never be a Nash equilibrium, given a free fee policy.

\end{abstract}
%
%
%
%------------------------------------------------
%---------------INTRODUCTION------------------------
%-----------------------------------------------
\section{Introduction}

%\paragraph{Payment Channels.} 
Distributed ledgers that employ the Nakamoto \cite{nakamoto2008bitcoin} or similar consensus mechanisms suffer from major scalability problems \cite{croman2016scaling}. In essence, the security of the consensus is based on the ability of each node to verify and store a replica of the entire blockchain history. %To avoid centralization, i.e., only a handful of companies being able to act as (full) nodes and guarantee the security of transactions, a limitation on the transaction throughput is enforced. 
Prominent solutions to this problem are payment channels  \cite{spilman2013channels,decker2015fast,poon2015lightning}. Payment channels are constructions that allow participants of the blockchain to execute the transactions off-chain while maintaining the security guarantees of the blockchain. The parties that enter a payment channel open a joint account with a specific capital and update the distribution of this capital  every time they exchange a transaction. This way the parties can execute an unlimited number of transactions as long as they all agree on the current distribution of capital. In case of a dispute, the blockchain acts as a judge, ensuring that the latest agreed capital distribution is enforced. 

%\paragraph{Payment Networks.} 
Multiple payment channels create a payment network, where a transaction can be executed even though there is no direct channel between payer and payee. Currently, this is achieved with the Bitcoin Lightning Network \cite{poon2015lightning} using Hash Timelock Contracts (HTLCs) \cite{decker2015fast,bitcoinwiki:htlcs}. 
When a sender wants to route a transaction through the network, a path of channels is discovered that has enough capital on each edge to route the transactions to the receiver \cite{prihodko2016flare,moreno2017silentwhispers,roos2017settling}. The intermediate nodes that move their capital act as service providers for the sender, and ask for a fee for their service. Therefore, operating a payment channel can offer revenue to the owner of the channel apart from reducing the cost of executing transactions on-chain. On the other hand, creating a payment channel is costly, since the opening and closing of the channel must occur on-chain, which requires to pay the regular blockchain transaction fee to the miner.

In this work, we study this trade-off. We investigate possible strategies for a participant. Assuming a constant blockchain fee for each transaction, when does it make sense to open a channel? If a path from sender to receiver already exists, does it make sense to create a cheaper path (or even a direct channel) to claim the fees for forwarding transactions? Our goal is to understand under which constraints specific network structures are Nash equilibiria. In other words, when can the participants of the network increase their profit (or decrease their costs) by changing the network structure? Further, we ask which network structures are stable under a free fee policy where each node on the network that operates a channel can set its own fee on the channel. To the best of our knowledge, our work is the first to analyze payment channels in a formal game theoretic model. 

\subsubsection{Our Contributions.}
First, we introduce a formal game-theoretic model which we later use to study the potential strategies for network creation for the nodes of the payment network. 
The model encapsulates a network creation game on payment channels: each node of the network can create multiple channels to other nodes to collect fees from the transactions that are routed through their channel. However, each channel creation costs the blockchain fee, and also each node competes with the other nodes in the network, since the sender of each transactions will choose the cheapest path to route the transaction to the receiver. 

We assume there is a fixed fee a sender has to pay to each intermediate node to route the transaction to the receiver. Further, we assume nodes have unlimited temporary capital, and that each pair of nodes are sender and receiver to an equal amount of transactions.
Under these assumptions, we explore various network structures, namely the path, the star, the complete bipartite graph and the clique. We find that each network structure constitutes a Nash equilibrium for a specific fee value. 
In particular, we show that the path is a weak Nash equilibrium only if the network fee is zero. Then, we show an upper bound for value of the fee on the star graph. The bound depends on the number of transactions, number of nodes and the value of the blockchain fee.
This means that creating a hub is a Nash equilibrium as long as the fee is very low  (compared to the blockchain fee). 
We observe that if the fee is above the upper bound, a two-stars structure emerges as a Nash equilibrium. We generalize this observation by examining complete bipartite graphs. We provide upper and lower bounds for such graphs which additionally depend on the number of centers (smaller side of the complete bipartite graph). This specific network structure defines an entire class of Nash equilibria. Finally, we consider the complete graph (clique), which is naturally a Nash equilibria when the network fee is very high.

From the plethora of network structures that constitute Nash equilibria when the fee policy is fixed, only the star is stable  under a free fee policy. Specifically, we show that the star is a pure Nash equilibrium, and the nodes will set the fee almost equal to the upper bound of the fixed fee policy. 
Further, we prove that the complete bipartite graph can never be a Nash equilibrium when nodes chose the fee of their channels freely as part of their strategy. 

% The contribution of the paper is summarized as follows:
% \begin{itemize}
%     \item We introduce the first game theoretic model for analyzing the network creation game on blockchain payment channels.
%     \item We examine various network structures (path, star, complete bipartite graph and clique) and determine for each one of them the constraints (fee value) under which they constitute a Nash equilibrium, given a fixed fee policy.
%     \item We determine that the star is a Nash equilibrium when each channel party can freely decide the channel fee. On the other hand, we prove the complete bipartite graph can never be a Nash equilibrium, given a free fee policy.
% \end{itemize}

%------------------------------------------------
%---------------MODEL------------------------
%-----------------------------------------------
\section{Model}

In this section, we introduce the game theoretic model. To this end, we first define the necessary notation and assumptions.

\paragraph{Capital \& fees.}
We assume all participants (nodes of the network) have unlimited temporary capital and thus the capital locked in all channels is also considered unlimited. This means that the channels can never be depleted. Moreover, this leads to stable fees that do not depend on the value of the routed transaction, because the participants are only interested in the number of transactions routed through their channel since the capital movement does not cost (they have unlimited capacity).

The cost of every  transaction and hence the cost of opening and closing a channel on the blockchain costs the blockchain fee.
We assume a fixed blockchain fee, $F_B\in\mathbb{R}^+ $, i.e., the fee is constant and stable over time. 
Further, we assume the fee is unilaterally paid by the node opening or closing the channel.
We also assume a fee $f_0$ for forwarding a transaction through a channel (that is not owned by the sender of the transaction) is the same for all nodes and stable in time.

%Moreover, we define as $s(p)$ and $r(p)$ the sender and receiver of transaction $p$, respectively. 

\paragraph{Information \& time.} We assume full information, i.e., every node of the system knows the complete payment scenario, and the channels created by other nodes. Although the decision of closing and opening a channel can occur at any time by any participant of the network, we assume a  simultaneous game, i.e., the participants open the channels in the beginning before executing any transactions. 
This is reasonable because we assume a full information game, thus every node knows its optimal strategy apriori.

\paragraph{Notation.}
The set of nodes participating in the network is denoted by $\bm{N}$, and the set of transactions to be executed (payment scenario) by $\bm{P}$. We assume $N>3$ and we use the terms transaction and payment interchangeably  throughout the paper. 
We define as $\bm{R}(x,p)$ the (set of) cheapest route(s) of a transaction $p \in \bm{P}$ in network state $x$. The network state $x$ is dependent on the strategy of the nodes, i.e. which channels the nodes have opened in the network. We note that if there is no route with cost lower than the blockchain fee, the set will return empty and the transaction will be executed on the blockchain.
%is this really needed? since we have simultaneous game why do we need a network state?? it is just the edges right??
Thus, the cost for a transaction $p$ for the sender of the transaction on a network state $x$ is the number of edges of the shortest path (for constant fee $f_0$ the cheapest is the shortest path) times the fee $f_0$,
$$f(x,p)=
\begin{cases}
\displaystyle (|\bm{R}(x,p)|-1) \cdot f_0 ,~if~\bm{R}(x,p) \ne \emptyset \\
F_B,~else
\end{cases} $$

On the other hand, the revenue of a node  $u\in\bm N$ from a transaction $p\in\bm P$ when executed on the network in state $x$ is $f_0$ if node $u$ is part of the cheapest route $\bm{R}(x,p)$; otherwise the revenue is zero.

Moreover, the \textit{strategy set} of a node $u\in\bm N$ is the set of all strategies available to node $u$. It is denoted as $\bm S_u$.
The strategy of node $u$ is denoted as $\mu_u\in\bm S_u$. The strategy set in our setting represents the channels a node decides to open at the beginning of the game.
A \textit{strategy combination} is a set containing a strategy for every node. The set of all possible strategy combinations is defined as $\bm S^N := \prod_{u\in\bm N} \bm S_u $, while a strategy combination is denoted by $\bm\mu\in\bm S^N $.
%is this part needed anywhere? with the strategies?
For simplicity, we will abuse the notation $\mu$ and $\mu_u$ to also denote the cardinality of the set, i.e., how many channels are open in the network and how many channels node $u$ opens in the network, respectively.
Last, we define the number of on-chain payments made by a node $u$ as $b_u$ and the total number of on-chain payments as $b=\sum_{\forall u \in \bm{N}} b_u$.

Next, we define the necessary functions for the analysis, namely the cost function, the social cost and the social optimum.

\paragraph{Cost Function.}
The cost function contains the cost for channel creation, on-chain payments, payments routed through the network and the revenue from forwarding payments. For a node $u$ it is defined as 
$$c(\bm{\mu}, \mu_u) = {\mu}_u
\cdot F_B + b_u \cdot F_B + \sum_{p \in \bm{P}: s(p)=u} f(x,p) -\sum_{p \in \bm{P}: u \in \bm{R}(x,p)} f_0$$ where $s(p)$ denotes the sender of transaction $p$. In our setting, sender and receiver of a transaction $p$ are the only relevant pieces of information of each payment, since all fees are independent of the value of a transaction.

\paragraph{Social Cost.}
The social cost (or negative welfare) is the sum of the costs of all nodes
$$-W = \sum_{n\in\bm{N}} c(\bm\mu,\mu_u)=
 (\mu + b) \cdot F_B $$

\paragraph{Social Optimum.}
The social optimum is the minimum social cost, which depends on the number of open channels and the number of payments executed on-chain. This term is minimized when all transactions are executed off-chain and the network forms a tree (connected with minimum number of channels). Hence, the social optimum is $min(-W) = (N-1)\cdot F_B$

%------------------------------------------------
%---------------GAMES------------------------
%-----------------------------------------------
\section{Channel Creation Game}
\label{chapter:sim}

First, we show some observations that hold generally under any set of transactions and graph structure. Then, we analyze specific structures and determine under which parameters they constitute a Nash equilibrium.

\subsection{Basic Properties}
\begin{lemma}\label{lem:double_channels}
In a pure Nash equilibrium, no channels are opened twice.
\end{lemma}
\begin{proof}
(Towards contradiction.) Suppose there is a pure  Nash equilibrium in which two nodes  have opened a channel with each other twice. In this case, each node can reduce his cost by not opening the second channel and thus the strategy cannot be a Nash equilibrium.  \hfill \qed
\end{proof}

\begin{lemma}\label{lem:onchain_payments}
In a pure, strict Nash equilibrium\footnote{If a strategy is always strictly better than all others for all profiles of other players' strategies, then it is strictly dominant. If the strategy is strictly dominant for all players, then it is a strict Nash equilibrium.}, none of the transactions are executed on-chain.
\end{lemma}
\begin{proof}
(Towards contradiction.) Suppose there is  a pure, strict Nash equilibrium in which a node (sender) executes a transaction on-chain. If there is a channel to the receiver of the transaction, the sender can reduce his cost by simply using the channel. The same holds if there is a path of channels with total fees less than the blockchain fee. Therefore, either there is no cheap path from sender to receiver or no path at all. 
In this case, if the sender has at least two transactions to send to the receiver, he would open a channel and reduce the cost. Thus, the sender has a single transaction to send to the receiver. However, the cost of opening a channel and the cost of executing the transaction on-chain is exactly the same (blockchain fee). If we assume there are no transactions routed through the channel from sender to receiver, the payoff (cost) of the sender in both strategies is the same. Hence, executing the transaction on-chain cannot be a strict Nash equilibrium. Thus, there are transactions routed through the channel from sender to receiver. Then, the payoff of the sender increases and hence the dominant strategy is to open the channel. This contradicts the assumption that executing the transaction on-chain is a Nash equilibrium.\hfill \qed
\end{proof}

Next, we analyse the channel creation game for a homogeneous payment scenario, where every node makes exactly $k \ge 1$ payments to every other node. The number of transactions is therefore $P = k \cdot N \cdot (N-1)$. 
For this payment scenario, we analyze multiple strategy combinations, i.e., different graph structures such as the path, the star, a complete bipartite graph and the clique, to discover under which parameters these graph structures constitute a Nash equilibrium.
We note that all graph structures that are trees (e.g. path, start, complete bipartite graph) are social optima.

%---------------path---------------------------
\subsection{Path}
\label{subsec:path}
The first graph structure we investigate is the path: each node connects to the node with the next higher ID. The node with the highest ID does not create a channel, but he is connected to the network through the node with the second highest ID. 

\paragraph{Social Cost.}
The social cost is $-W = (N - 1) \cdot F_B$.

\paragraph{Nash Equilibrium.}
A specific strategy is a Nash equilibrium if none of the players can increase their payoff by deviating from it. This means that the path is a NE only if all possible deviations lead to higher cost (or equivalent if it is a weak NE) for the deviating node. 
We observe that for  $f_0=0$, deviating from the path structure cannot decrease the cost, thus the path is a weak NE (similarly to every other tree structure).
However, for any fee $f_0>0$, the first node can increase the revenue from the fees, and thus decrease his cost, by connecting to a middle node on the path and allowing transactions to be routed through his channel. Thus, for any non-zero fee, the path is not a NE.

%---------------star---------------------------
\subsection{Star}
\label{subsec:star}
The second graph structure we investigate is the star: one node creates channels to everyone else, while the other nodes do not create any channels.  The strategy to create channels to $a\in[0,N-1]$ outer nodes is denoted as $(a)$.

\paragraph{Cost Functions.}
The cost of the center node is
$c(\bm\mu,(N-1)) = (N-1) \cdot F_B - (N-1) \cdot (N-2) \cdot k \cdot f_0 $.\\
The cost of the outer nodes is
$c(\bm\mu,(0)) = (N-2) \cdot k \cdot f_0 $.\\
The social cost is
$-W=(N-1) \cdot F_B$.

\paragraph{Nash Equilibrium.}
A node can only deviate from the strategy as follows: An outer node creates channels to $a\in[1,N-2]$ other outer nodes. This holds because in the homogeneous payment scenario, a pure Nash equilibrium demands a connected graph, else Lemma~\ref{lem:onchain_payments} is violated. This means the center node will not disconnect the graph. Moreover, from Lemma~\ref{lem:double_channels} no outer node will create a channel to the center node.

If an outer node creates channels to $a\in[1,N-2]$ other outer nodes, his cost function is $c(\bm\mu,(a)) = a \cdot F_B + (N - 2 - a) \cdot k \cdot f_0 - a \cdot (a-1) \cdot k \cdot \frac{1}{2} \cdot f_0$.

If there is an $a\in[1,N-2]$ for which the cost function of an outer node is decreased, then the star is not a NE. 
Since the second derivative with respect to $a$ is strictly negative, we only have to check the corner cases below:

\begin{itemize}
    \item For $a=1$: $c(\bm\mu,(0)) < c(\bm\mu,(1)) \iff f_0 < \frac{F_B}{k} $.
    \item For $a=N-2$: 
$c(\bm\mu,(0)) < c(\bm\mu,(N-2)) \iff f_0 < \frac{F_B}{k} \cdot \frac{2}{N-1} $.
\end{itemize}

Thus, for $N>3$, the star is a NE when 
\[f_0 < \frac{F_B}{k} \cdot \frac{2}{N-1} \]

%---------------2 stars---------------------------
\subsection{Star with two centers}
\label{subsec:2star}
We observe that for a very low constant fee any star can be a Nash equilibrium. Further, we notice that if the fee is high enough the dominant strategy for the outer nodes is to create more channels, eventually becoming the center of a second star. We examine this exact case, where there are two center nodes, each creating channels to all outer nodes, but not to each other. The outer nodes do not create any channels. We denote the strategy to create $a\in[0,2]$ channels to center nodes and $b\in[0,N-2]$ channels to outer nodes as $(a,b)$.

\paragraph{Cost Functions.}
The cost of the center nodes is
$c(\bm\mu,(0,N-2)) = (N-2) \cdot F_B + k \cdot f_0 - (N-2) \cdot (N-3) \cdot k \cdot \frac{1}{2} f_0 $.\\
The cost of the outer nodes is
$c(\bm\mu,(0,0)) = (N-3) \cdot k \cdot f_0 - 2 \cdot k \cdot \frac{1}{(N-2)} f_0 $.\\
The social cost is
$-W = 2 \cdot (N-2) \cdot F_B $.

\paragraph{Nash Equilibrium.}
The nodes can deviate as follows:
\begin{enumerate}[label=(\Alph*)]
    \item  A center node creates channels to $b\in[1,N-3]$ outer nodes.
    \item  A center node creates channels to $b\in[0,N-2]$ outer nodes and creates a channel to the other center node.
    \item  An outer node creates channels to $b\in[1,N-3]$ other outer nodes.
\end{enumerate}
We analyze each case to determine the parameter space for which these deviations do not decrease the cost of a node, and thus the two center structure is a NE.

\paragraph{\textbf{(Deviation A)}}
If a center node created channels to only $b\in[1,N-3]$ outer nodes, his cost function would become
\[c(\bm\mu,(0,b)) = b \cdot F_B + k \cdot f_0 + (N-2-b) \cdot k \cdot 2f_0 - b \cdot (b-1) \cdot k \cdot \frac{1}{2} f_0 \]
This cost function must be higher than $c(\bm\mu,(0,N-2))$ for all $b$. Since the second derivative w.r.t. $b$ is strictly negative, we only have to check the corner cases:
\begin{itemize}
    \item For $b=1$: $c(\bm\mu,(0,N-2)) < c(\bm\mu,(0,1)) \iff f_0 > \frac{F_B}{k} \cdot \frac{2}{N+2} $.
    \item For $b=N-3$: $c(\bm\mu,(0,N-2)) < c(\bm\mu,(0,N-3)) \iff f_0 > \frac{F_B}{k} \cdot \frac{1}{N-1} $
\end{itemize}

\paragraph{\textbf{(Deviation B)}}
If an center node creates a channel to the other center node and channels to $b\in[0,N-2]$ outer nodes, his cost function is
\[c(\bm\mu,(1,b) = (b+1) \cdot F_B + (N-2-b) \cdot k \cdot f_0 - b \cdot (b-1) \cdot k \cdot \frac{1}{2} f_0 \]
This cost function must be higher than $c(\bm\mu,(0,N-2))$ for all $b$. Similarly, we only have to check the corner cases:
\begin{itemize}
    \item For $b=0$: $c(\bm\mu,(0,N-2)) < c(\bm\mu,(1,0)) \iff f_0 < \frac{F_B}{k} \cdot \frac{2}{N}  $.
    \item For $b=N-2$: $c(\bm\mu,(0,N-2)) < c(\bm\mu,(1,N-2)) \iff f_0 > \frac{F_B}{k}$.
\end{itemize}

\paragraph{\textbf{(Deviation C)}}
If an outer node creates channels to $b\in[1,N-2]$ other outer nodes, his cost function is
\[c(\bm\mu,(0,b)) = b \cdot F_B + (N-3-b) \cdot k \cdot f_0 - 2 \cdot k \cdot \frac{1}{(N-2)} f_0 - b \cdot (b-1) \cdot k \cdot \frac{1}{3}f_0 \]
This cost function must be higher than $c_o(\bm\mu,(0,0))$ for all $b$. Similarly, we only have to check the corner cases:
\begin{itemize}
    \item For $b=1$: $c(\bm\mu,(0,0)) < c(\bm\mu,(0,1)) \iff f_0 < \frac{F_B}{k} $.
    \item For $b=N-3$: $c(\bm\mu,(0,0)) < c(\bm\mu,(0,N-3)) \iff f_0 < \frac{F_B}{k} \cdot \frac{3}{N-1}$.
\end{itemize}

Combining all the bounds from the deviating strategies, for $N>3$, we derive the parameter space for which the two center structure is a NE. Specifically, the conditions reduce to 
\[\frac{F_B}{k} \cdot \frac{2}{N} < f_0 < \frac{F_B}{k} \cdot \frac{3}{N-1} \]

\subsection{Complete bipartite graph}
\label{subsec:bipartite}
Previously, we showed that stars with one or two center nodes can be a Nash equilibrium, if there is a constant fee that fulfills certain conditions. 
Furthermore, if $f_0$ is high, outer nodes decrease their cost by creating channels to other outer nodes. Intuitively, this leads to a NE that is a bipartite graph structure.
We study exactly this case:
$c\in[2,\sfrac{N}{2}]$ nodes build a center by creating channels to everyone else but each other. 
For simplicity we denote the number of outer nodes by $d:=N-c$. 
The network structure now is a complete bipartite graph with $c$ nodes in the smaller partition and $d$ nodes in the larger partition. We denote the strategy to create channels to $a\in[0,c]$ center nodes and to $b\in[0,d]$ outer nodes as $(a,b)$.

\paragraph{Cost Functions.}
The cost of the center nodes is
$c(\bm\mu,(0,d)) = d \cdot F_B + (c-1) \cdot k \cdot f_0 - d \cdot (d-1) \cdot k \cdot \frac{1}{c} f_0 $.\\
The cost of the outer nodes is 
$c(\bm\mu,(0,0)) = (d-1) \cdot k \cdot f_0 - c \cdot (c-1) \cdot k \cdot \frac{1}{d} f_0 $.\\
The social cost is
$-W = c \cdot (N-c) \cdot F_B $.

\paragraph{Nash Equilibrium.}
The nodes can deviate as follows:
\begin{enumerate}[label=(\Alph*)]
    \item  A center node creates channels to only $b\in[1,d-1]$ outer nodes.
    \item  A center node creates channels to $a\in[1,c-1]$ center nodes and to $b=0$ outer nodes.
    \item  A center node creates channels to $a\in[1,c-1]$ center nodes and to $b\in[1,d]$ outer nodes.
    \item  An outer node creates channels to $b\in[1,d-1]$ other outer nodes.
\end{enumerate}
Next, we discover the parameter space for which the strategies above lead to the increase of the cost function of a node. 

\paragraph{\textbf{(Deviation A)}}
If a center node created channels to only $b\in[1,d-1]$ outer nodes, his cost function would become
\[c(\bm\mu,(0,b)) = b \cdot F_B + (c-1) \cdot k \cdot f_0 + (d-b) \cdot k \cdot 2f_0 - b \cdot (b-1) \cdot k \cdot \frac{1}{c} f_0 \]
This cost function must be higher than $c(\bm\mu,(0,d))$  for all $b$. Since the second derivative w.r.t. $b$ is strictly negative, we only  check the corner cases:
\begin{itemize}
    \item For $b=1$: $c(\bm\mu,(0,d)) < c(\bm\mu,(0,1)) \iff \frac{F_B}{k} \cdot \frac{c}{N+c} < f_0  $.
        \item For $b=d-1$: $c(\bm\mu,(0,d)) < c(\bm\mu,(0,d-1))  \iff \frac{F_B}{k} \cdot \frac{c}{2N-2} < f_0$.
\end{itemize}

\paragraph{\textbf{(Deviation B)}}
If a center node creates channels to $a\in[1,c-1]$ other center nodes and to $b\in[1,d]$ outer nodes, his cost function is
\[c(\mu\bm,(a,b)) = (a+b) \cdot F_B + (d-b) \cdot k \cdot f_0 + (c-1-a) \cdot k \cdot f_0 \]
\[ - b \cdot (b-1) \cdot k \cdot \frac{1}{c} f_0 - a \cdot (a-1) \cdot k \cdot \frac{1}{d+1} f_0 \]
This cost function must be higher than $c(\bm\mu,(0,d))$. Since the second derivatives w.r.t. $a$ and $b$ are strictly negative, we only have to check the corner cases:
\begin{itemize}
    \item For $a=1,b=1$: \[c(\bm\mu,(0,d)) < c(\bm\mu,(1,1)) \iff \frac{F_B}{k} \cdot \frac{cN-c^2-2c}{N^2 - cN + N - 3c} < f_0\]
    \item For $a=1,b=d$: $c(\bm\mu,(0,d)) < c(\bm\mu,(1,d)) \iff f_0 < \frac{F_B}{k}$
    \item For $a=c-1,b=1$: \[c(\bm\mu,(0,d)) < c(\bm\mu,(c-1,1)) \iff \frac{F_B}{k} \cdot \frac{cN^2 - 3c^2N + cN + 2c^3 - 2c^2}{N^3 - 2cN^2 + cN - N + c^2 - c} <  f_0\]
    \item For $a=c-1,b=d$: $c(\bm\mu,(0,d)) < c(\bm\mu,(c-1,d)) \iff f_0 < \frac{F_B}{k} \cdot \frac{N-c+1}{N-1} $
\end{itemize}

\paragraph{\textbf{(Deviation C)}}
If a center node creates channels to $a\in[1,c-1]$ other center nodes and to $b=0$ outer nodes, his cost function is
\[c(\bm\mu,(a,0)) = a \cdot F_B + d \cdot k \cdot f_0 + (c-1-a) \cdot k \cdot 2f_0 - (a) \cdot (a-1) \cdot k \cdot \frac{1}{d+1} f_0 \]
This cost function must be higher than $c(\bm\mu,(0,d))$ for all $a$. Since the second derivative w.r.t. $a$ is strictly negative, we only have to check the corner cases:
\begin{itemize}
    \item For $a=1$: \[c(\bm\mu,(0,d)) < c(\bm\mu,(1,0)) \iff \frac{F_B}{k} \cdot \frac{cN-c^2-c}{N^2-cN-N+c^2-2c} < f_0 \]
    \item For $a=c-1$: \[c(\bm\mu,(0,d)) < c(\bm\mu,(c-1,0))  \iff \frac{F_B}{k} \cdot \frac{cN^2 - 3c^2N + 2cN + 2c^3 - 3c^2 + c}{N^3 - 2cN^2 + 2cN - N} < f_0\]
\end{itemize}

\paragraph{\textbf{(Deviation D)}}
If an outer node creates channels to $b\in[1,d-1]$ other outer nodes, his cost function is
\[c(\bm\mu,(0,b)) = b \cdot F_B + (d-1-b) \cdot k \cdot f_0 - c \cdot (c-1) \cdot k \cdot \frac{1}{d} f_0 - b \cdot (n-1) \cdot k \cdot \frac{1}{c+1}f_0 \]
This cost function must be higher than $c_o(\bm\mu,(0,0))$ for all $b$. Since the second derivative w.r.t. $b$ is strictly negative, we only have to check the corner cases:
\begin{itemize}
    \item For $b=1$: $c(\bm\mu,(0,0)) < c(\bm\mu,(0,1)) \iff f_0 < \frac{F_B}{k} $
    \item For $b=d-1$: $c(\bm\mu,(0,0)) < c(\bm\mu,(0,d-1)) \iff f_0 < \frac{F_B}{k} \cdot \frac{c+1}{N-1}$
\end{itemize}

% From the analysis of the alternative strategies, we derive upper and lower bounds for the value of $f_0$ for which this strategy combination is a Nash equilibrium:
% \begin{itemize}
%     \item $f_0 > \frac{F_B}{k} \cdot \frac{c}{2N-2} $ \hfill \textit{(lower bound 1)}
%     \item $f_0 > \frac{F_B}{k} \cdot \frac{c}{N+c} $ \hfill \textit{(lower bound 2)}
%     \item $f_0 > \frac{F_B}{k} \cdot \frac{cN-c^2-2c}{N^2 - cN + N - 3c} $ \hfill \textit{(lower bound 3)}
%     \item $f_0 > \frac{cN^2 - 3c^2N + cN + 2c^3 - 2c^2}{N^3 - 2cN^2 + cN - N + c^2 - c} $ \hfill \textit{(lower bound 4)}
%     \item $f_0 > \frac{F_B}{k} \cdot \frac{cN-c^2-c}{N^2-cN-N+c^2-2c} $ \hfill \textit{(lower bound 5)}
%     \item $f_0 > \frac{F_B}{k} \cdot \frac{cN^2 - 3c^2N + 2cN + 2c^3 - 3c^2 + c}{N^3 - 2cN^2 + 2cN - N} $ \hfill \textit{(lower bound 6)}
%     \item $f_0 < \frac{F_B}{k} $ \hfill \textit{(upper bound 1)}
%     \item $f_0 < \frac{F_B}{k} \cdot \frac{N-c+1}{N-1} $ \hfill \textit{(upper bound 2)}
%     \item $f_0 < \frac{F_B}{k} \cdot \frac{c+1}{N-1} $ \hfill \textit{(upper bound 3)}
% \end{itemize}

From the analysis on the possible deviations from the strategy, we derive multiple upper and lower bounds for the value of $f_0$. For $N>3$ and $2\leq c \leq \sfrac{N}{2}$ these conditions reduce to the following:
\[f_0 > \frac{F_B}{k} \cdot \frac{cN-c^2-2c}{N^2 - cN + N - 3c} \]
\[f_0 > \frac{F_B}{k} \cdot \frac{cN-c^2-c}{N^2-cN-N+c^2-2c} \]
\[f_0 < \frac{F_B}{k} \cdot \frac{c+1}{N-1} \]

We have defined not only one Nash equilibrium in this analysis, but a whole class of Nash equilibria. Table~\ref{tab:sim:bounds} shows the numerical values for the bounds of a complete bipartite graph as Nash equilibrium. Figure~\ref{fig:bounds_plot} shows a plot of the bounds for $N=10^3$. The Nash equilibria lay in the thin area between the lowest red and the highest blue line.

\begin{table}[ht]
    \centering
    \begin{tabular}{c|c|c|c|c|c}
    N & c & lower bound [$\frac{F_B}{k}$] & upper bound [$\frac{F_B}{k}$] & active lb & active ub \\
    \hline
    $10^3$  & 2     & $.2000000\cdot10^{-2}$   & $.30030\cdot10^{-2}$  & 5 & 3 \\
            & 3     & $.2999991\cdot10^{-2}$   & $.40040\cdot10^{-2}$  & 5 & 3 \\
            & 5     & $.4999925\cdot10^{-2}$   & $.60060\cdot10^{-2}$  & 5 & 3 \\
            & 10    & $.9999192\cdot10^{-2}$   & $.11011\cdot10^{-1}$  & 5 & 3 \\
            & 100   & $.9970024\cdot10^{-1}$   & $.10110$               & 3 & 3 \\
            & 499   & $.4975016$                & $.50050$               & 3 & 3 \\
            & 500   & $.4984984$                & $.50150$               & 3 & 3 \\
    \hline
    $10^4$  & 2     & $.20000\cdot10^{-3}$     & $.30003\cdot10^{-3}$  & 5 & 3 \\
            & 3     & $.29997\cdot10^{-3}$     & $.40004\cdot10^{-3}$  & 5 & 3 \\
            & 5     & $.49995\cdot10^{-3}$     & $.60006\cdot10^{-3}$  & 5 & 3 \\
            & 10    & $.99990\cdot10^{-3}$     & $.11001\cdot10^{-2}$  & 5 & 3 \\
            & 100   & $.99981\cdot10^{-2}$     & $.10101\cdot10^{-1}$  & 5 & 3 \\
            & 1000  & $.99971\cdot10^{-1}$     & $.10011$               & 3 & 3 \\
            & 4999  & $.49976$                  & $.50005$               & 3 & 3 \\
            & 5000  & $.49985$                  & $.50015$               & 3 & 3 \\
    \hline
    $10^5$  & 2     & $.200000\cdot10^{-4}$    & $.300003\cdot10^{-4}$ & 5 & 3 \\
            & 3     & $.299997\cdot10^{-4}$    & $.400004\cdot10^{-4}$ & 5 & 3 \\
            & 5     & $.499995\cdot10^{-4}$    & $.600006\cdot10^{-4}$ & 5 & 3 \\
            & 10    & $.999990\cdot10^{-4}$    & $.110001\cdot10^{-3}$ & 5 & 3 \\
            & 100   & $.999990\cdot10^{-3}$    & $.101001\cdot10^{-2}$ & 5 & 3 \\
            & 1000  & $.999971\cdot10^{-2}$    & $.100101\cdot10^{-1}$ & 3 & 3 \\
            & 10000 & $.999971\cdot10^{-1}$    & $.100011$              & 3 & 3 \\
            & 49999 & $.499976$                 & $.500005$              & 3 & 3 \\
            & 50000 & $.499985$                 & $.500015$              & 3 & 3 \\
    \end{tabular}
    \caption{Numerical results for the lower and bounds for a complete bipartite graph as a Nash equilibrium.}
    \label{tab:sim:bounds}
\end{table}

\begin{figure}[ht]
\centering
\begin{minipage}{.5\textwidth}
  \centering
    \includegraphics[width=0.9\textwidth]{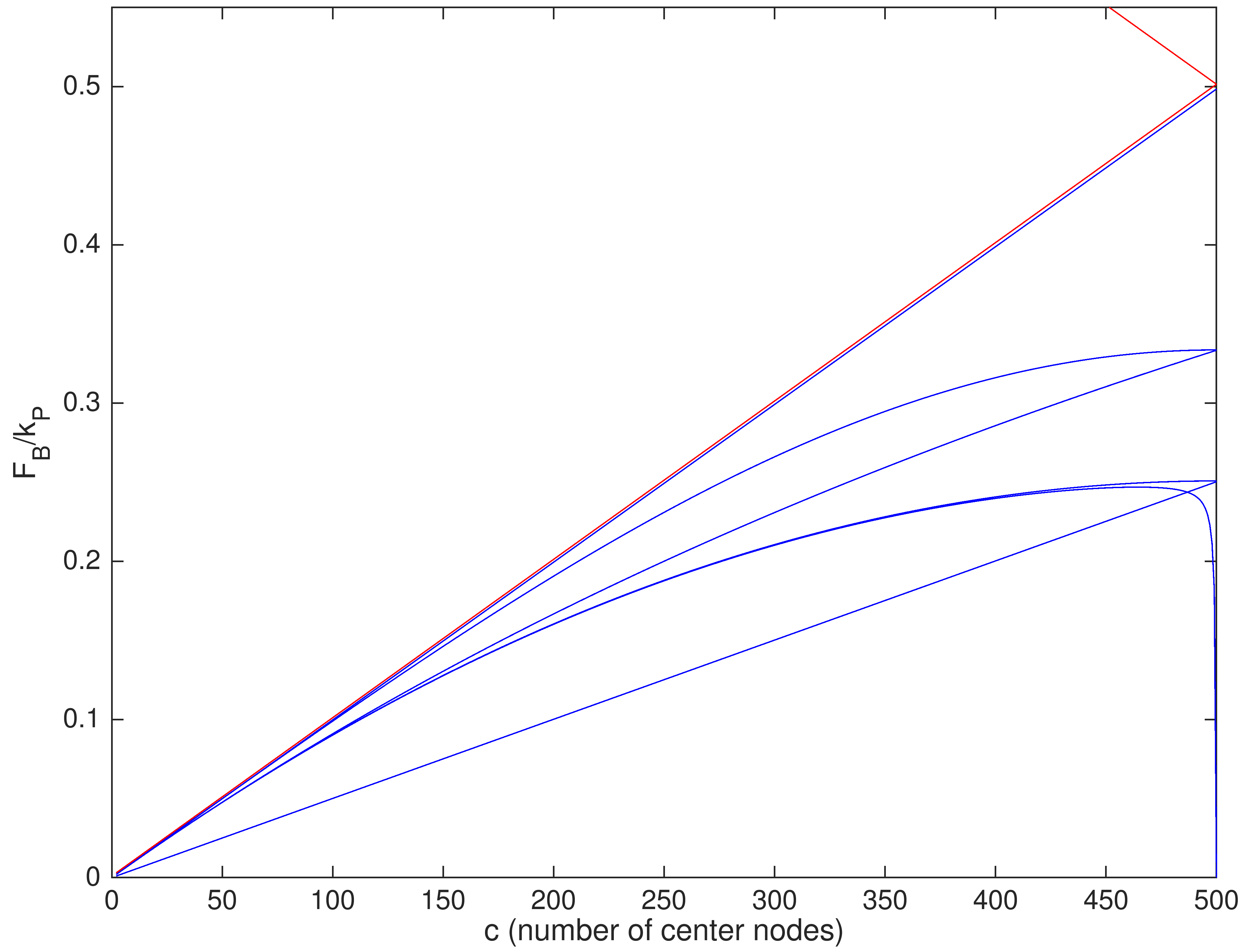}
\end{minipage}%
\begin{minipage}{.5\textwidth}
  \centering
    \includegraphics[width=0.9\textwidth]{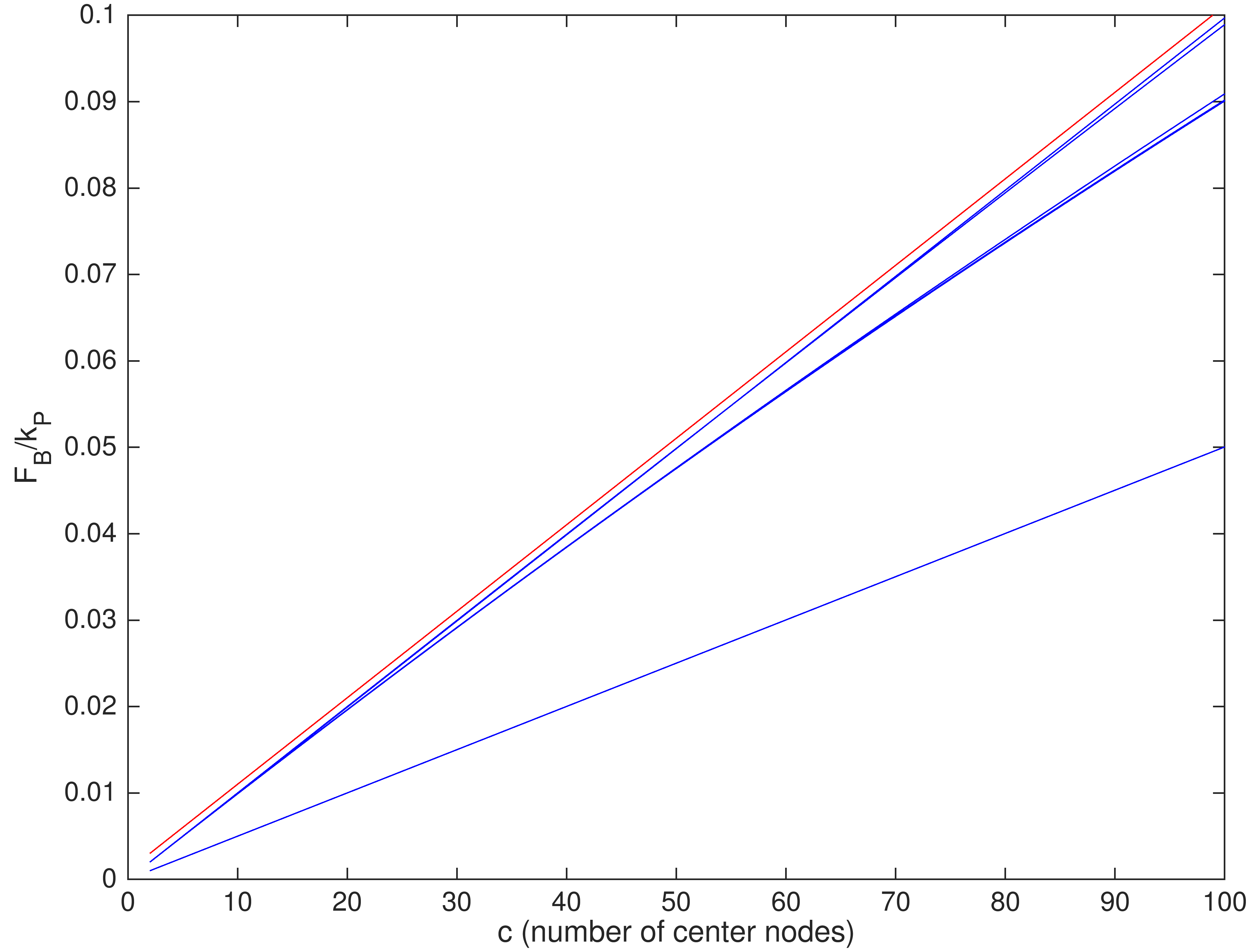}
\end{minipage}
\caption{Plots of the bounds for $N=10^3$ with the upper bounds in red and the lower bounds in blue.}
\label{fig:bounds_plot}
\end{figure}

% \begin{figure}[H]
%     \centering
%     \includegraphics[width=0.5\textwidth]{bounds_master.png}
%     \caption{Plot of the bounds for $N=10^3$ with the upper bounds in red and the lower bounds in blue}
%     \label{fig:bounds_plot_master}
% \end{figure}

% \begin{figure}[H]
%     \centering
%     \includegraphics[width=0.5\textwidth]{bounds_detail.png}
%     \caption{Zoomed in plot of the bounds for $N=10^3$ with the upper bounds in red and the lower bounds in blue}
%     \label{fig:bounds_plot_detail}
% \end{figure}

\subsection{Clique}
\label{subsec:clique}
The last graph structure we investigate is the clique, i.e., the complete graph. In this case, the $i$-th node opens $N-i$ channels. The strategy of creating channels to $a$ nodes without self-loops is denoted as $(a)$.

\paragraph{Cost Functions.}
The cost of the $i$-th node is $c(\bm\mu,(N-i)) = (N-i) \cdot F_B$.
The social cost is $-W = \frac{N \cdot (N - 1)}{2} \cdot F_B$.

\paragraph{Nash Equilibrium.}
The nodes can deviate as follows:
\begin{enumerate}[label=(\Alph*)]
    \item The first node creates channels to only $a\in[1,N-2]$ other nodes.
    \item Node $i$ (but not the first or last one) creates channels to only $a\in[0,N-i-1]$ nodes from the set of nodes he would originally connect to (node $i+1$ to node $N$).
\end{enumerate}
Now, we analyze these deviation strategies to explore the parameter space for which the strategies above lead to the increase of the cost function of a node. 

\paragraph{\textbf{(Deviation A)}}
If the first node creates channels to $a\in[1,N-2]$ other nodes his cost function is
$c(\bm\mu,(a)) = a \cdot F_B + (N-1-a) \cdot k \cdot f_0 $.
This cost function must be higher than $c(\bm\mu,(N-1))$ for all $a$.
Thus, 
$c(\bm\mu,(N-1)) < c(\bm\mu,(a)) 
\iff f_0 > \frac{F_B}{k} $.

\paragraph{\textbf{(Deviation B)}}
If node $i$ (not the first or last one) creates channels to $a\in[0,N-i-1]$ nodes from the set of nodes he would originally connect to (node $i+1$ to node $N$), his cost functions is
$c(\bm\mu,(a)) = a \cdot F_B + (N-i-a) \cdot k \cdot f_0 $.
This cost function must be higher than $c(\bm\mu,(N-i))$ for all $a$.
Thus,
$c(\bm\mu,(N-i)) < c(\bm\mu,(a)) 
\iff f_0 > \frac{F_B}{k} $.

\paragraph{}To summarize, the clique is a Nash equilibrium for $f_0 > \frac{F_B}{k}$.

% \subsection{Conclusion}
% We assumed a homogeneous payment scenario and a globally defined network fee. With this assumption any complete bipartite graph with the smaller group of nodes creating the channels can be a Nash equilibrium for certain conditions on the network fee.
% In particular, for high network fees ($0.5F_B < f_0 < F_B$), we observe that the nodes start creating channels to the other nodes in their subgroup until there is a clique for very high fees.

\subsection{The fee game}
% \subsubsection{The fee game.}
\label{sec:sim:fee_game}
In this subsection we investigate how the nodes set the fees, if the network structure is fixed to one of the previously found Nash equilibria. Especially, we try to find a Nash equilibrium, which also holds for the conditions discussed in the previous subsections. Therefore, we slightly change the model as follows: $x$ is a Nash equilibrium from the previous subsections. The nodes can only set a constant fee on each of their channels. Further, the nodes cannot create new channels. We still consider simultaneous game play, i.e., the nodes must (simultaneously) choose their strategy before any transaction is executed in the network. The payment scenario is still the same (homogeneous).

\paragraph{Complete Bipartite Graph.}
We start with a complete bipartite graph with $c\in[2,\sfrac{N}{2}]$ nodes in the smaller partition creating the channels. The goal of this analysis is to gain a better intuition of the fee evolution and therefore determine the number of hubs that will be eventually established.

We make following statement about the Nash equilibrium of the described game.

\begin{lemma}%[Nash Equilibrium of the Fee Game, Complete Bipartite Graph]
\label{lem:fees}
For $k>2$, a strategy combination is a (weak) Nash equilibrium if and only if there are at least two node-distinct paths which are free (zero fees) for all indirect transactions.
\end{lemma}
\begin{proof}
$(\rightarrow)$ Suppose there is a Nash equilibrium in which there is a set of $k$ indirect transactions (according to the payment scenario, each pair of sender and receiver executes $k$ transactions) that can be routed only through a single path with positive fees. Then, another node will open a channel to connect the two nodes and increase his payoff, since on expectation half the $k$ transactions will go through the new channel. Contradiction. 

Suppose now, there is only a single path with zero fees to route the transaction. Then, the nodes acting as intermediaries will increase the fee almost matching the price of the blockchain fee. In such a case, we have the same effect as described above. Contradiction.

Therefore, there cannot be a single path connecting any two nodes in the network.
Suppose now there are multiple node-distinct paths that connect sender and receiver, but with positive fees. Then, each node of these paths will decrease the fee in an attempt to win out the competition by being the cheapest path. Contradiction.
Therefore, any strategy that is a Nash equilibrium must contain at least two node-distinct paths for each pair of sender and receiver. 

$(\leftarrow)$ If there exist a path with zero fees for every transaction, no node stands to gain from opening a new channel. Furthermore, increasing the fee in any path will not lead to the decrease of the cost function (higher revenue) since the path containing the non-zero fee will be ignored and no transactions will be routed through such a path.
Thus, no node can gain from choosing a different strategy, i.e., the strategy combination is a Nash equilibrium. \hfill \qed
\end{proof}

\begin{theorem}
The complete bipartite graph is not a Nash Equilibrium when the nodes are free to chose the fees on their channels.
\end{theorem}
\begin{proof}
Follows immediately from Lemma \ref{lem:fees} and the lower bound established in subsection \ref{subsec:bipartite}.\hfill \qed
\end{proof}

% As we see the fees must be zero for the proposed network structure to be a NE. But this contradicts the conditions for the complete bipartite graph with at least two nodes on each side. This leads to the assumption that the bipartite graph is reduced to a star ($c=1$).

\paragraph{Star.}
Next, we consider the fee evolution when the network structure is a star. The reason we proceed with this specific network structure is the previous observations; when nodes can freely chose the fees they impose on their channels, having multiple paths leads to zero-fee paths. Intuitively, the star does not suffer from this problem.

Particularly, we notice that the center node is the only node that can charge fees, since all transactions are routed through the center node. However, in subsection \ref{subsec:star}, we showed that there is an upper bound on the value of the fee the center node can ask for; otherwise other nodes will deviate from the strategy combination and form a second hub. 
% Thus, the star is a pure Nash equilibrium when $f_0 = \frac{F_B}{k} \cdot \frac{2}{N-1} - \epsilon$.

\begin{corollary}
The star is a pure Nash equilibrium when $f_0 = \frac{F_B}{k} \cdot \frac{2}{N-1} - \epsilon$.
\end{corollary}

\section{Related Work}\label{chapter:relw}

Payment channels were originally introduced by Spilman~\cite{spilman2013channels}. The core idea was to use unidirectional channels with a predefined sender and receiver. Later, various constructions for bidirectional payment channels were proposed~\cite{spilman2013channels,decker2015fast,poon2015lightning,decker2018eltoo,avarikioti2019brick}. They all use a common account for the parties and off-chain exchange of signed transactions proving the state of the channel. The creation of multiple such channels on a common blockchain network leads to the formation of channel networks, such as the Lightning network~\cite{poon2015lightning} on Bitcoin~\cite{nakamoto2008bitcoin}, and the Raiden network \cite{raiden2017} on Ethereum \cite{wood2014ethereum}. In this work, we study different strategies for the nodes in such payment  networks, independent of which payment channel construction method is used. Thus, this work applies to all payment channel solutions.

Avarikioti~et~al.~\cite{avarikioti2018algorithmic,avarikioti2018payment} formulated a similar problem to the one studied in this paper.  Their goal was to find an optimal strategy for a central coordinator, a so-called payment service provider. In contrast to  \cite{avarikioti2018algorithmic,avarikioti2018payment}, our work studies a situation with \textit{multiple} players. In other words, our work is rooted in the area of game theory, whereas  \cite{avarikioti2018algorithmic,avarikioti2018payment} was using optimization methods. Despite these differences,  
we found that the near-optimal solution of \cite{avarikioti2018algorithmic,avarikioti2018payment} (the star as network structure) is also a Nash equilibrium in an uncoordinated situation. So we get a similar result despite two completely different approaches. This is a strong indication that the Lightning network (and similar others) will eventually develop into a more centralized network structure.

Network creation games, originally introduced in by Fabrikant et al.~\cite{fabrikant2003network}, are used to model distributed networks with rational players. Each player wants to maximize/minimize a profit/cost function which represents the cost of creating and using the network.
Fabrikant et al.~\cite{fabrikant2003network} modeled the Internet using Network Creation Games. They introduced a cost function containing the network creation cost and the sum of the distances to the other nodes. For their model, they proved upper and lower bounds for the Price of Anarchy(PoA). They also conjectured that the Nash equilibria in this game are trees, however this was disproved by Albers et al.~\cite{albers2014nash}. Alon et al.~\cite{alon2010basic} aimed for stronger bounds on PoA of the Network Creation Game (sum and local-diameter version). 
Both these works, however, use simple cost functions, where the creation cost and the usage cost of an node are independent. %There are many different variants of network creation games, including some that deal with creating overlay networks, e.g. \cite{Moscibroda2006on}.
In contrast, in this work the cost function of a node on the payment network contains both the revenue from the fees of the channel when the node is an intermediate node in a multihop transaction as well as the fees paid by the node when he is sending a multihop transaction. Hence, the cost function depends on the state of the network which itself contains the individual fee policies of the nodes. Overall, the channel creation game is probably more complex than previous work in this domain.

%------------------------------------------------
%---------------CONCLUSION------------------------
%-----------------------------------------------
\section{Conclusion}
We introduced the first game-theoretic model that encapsulates the payment channel creation game on blockchain networks. 
First, we explored various network structures and determined the parameter space  for which they constitute a Nash equilibrium. 
For the analysis, we initially assumed a fixed fee policy where each node benefits the same for each transaction routed through any of his channels. Then, we briefly considered a free fee policy, where the fee of each channel is part of the strategy of the node.

Particularly, for the fixed fee policy, we observed that the path is a Nash equilibrium only when the fee is zero. Otherwise, for a small positive fee we noticed the formation of a star. Furthermore, we showed that beyond an upper bound the star ceased to be a Nash equilibrium and multiple star structures emerged. This observation lead to the investigation of the complete bipartite graph which defined a class of Nash equilibria dependent on the correlation between the sizes of the two independent sets of nodes. Finally, the complete graph was proven to be a Nash equilibrium when the fee is relatively high, as expected. 

More importantly, we showed that even in a free fee policy, the star with uniform fees almost equal to the upper bound discussed above is a pure Nash equilibrium. 
On the contrary, the complete bipartite graph was proven unstable under this fee policy; we proved that a complete bipartite graph can never be a Nash equilibrium.
We note, that these observations indicate the stability of a star structure, even though its centralized nature is opposed to the philosophy of decentralized and distributed payment networks.

 \bibliographystyle{splncs04}
 \bibliography{ref}

\end{document}